\documentclass[11pt,letterpaper]{article}
\usepackage{color,ifpdf,latexsym,graphicx,url}
\usepackage[margin=1in]{geometry}
\usepackage{palatino,mathpazo}
\usepackage{amsmath, amssymb, amsthm, amsfonts,alltt,clrscode}
\usepackage{array,colortbl,hhline,xcolor}
\usepackage{subcaption}

\title{Traversability, Reconfiguration, and Reachability in the Gadget Framework}
\author{Joshua Ani, Erik Demaine, Yevhenii Diomidov, Dylan Hendrickson, Jayson Lynch}

\usepackage{hyperref}
\hypersetup{breaklinks,bookmarksnumbered,bookmarksopen,bookmarksopenlevel=2}
{\makeatletter \hypersetup{pdftitle={\@title}}}

{\makeatletter \gdef\fps@figure{!htbp}}

\let\realbfseries=\bfseries
\def\bfseries{\realbfseries\boldmath}

\renewcommand\emph[1]{\textbf{\textit{\boldmath #1}}}

\newtheorem{theorem}{Theorem}[section]
\newtheorem{lemma}[theorem]{Lemma}
\newtheorem{corollary}[theorem]{Corollary}
\newtheorem{definition}[theorem]{Definition}
\newtheorem{problem}[theorem]{Problem}

\newcommand{\rdni}{non-interacting box gadget}

\graphicspath{{figures/}}

\begin{document}

\maketitle

\begin{abstract}
Consider an agent traversing a graph of ``gadgets'', each with local state
that changes with each traversal by the agent.
We characterize the complexity of \emph{universal traversal}, where the goal
is to traverse every gadget at least once, for DAG gadgets, one-state gadgets,
and reversible deterministic gadgets.
We also study the complexity of \emph{reconfiguration},
where the goal is to bring the system of gadgets to a specified state,
proving many cases PSPACE-complete, and showing in some cases that
reconfiguration can be strictly harder than \emph{reachability}
(where the goal is for the agent to reach a specified location),
while in other cases, reachability is strictly harder than reconfiguration.
\end{abstract}

\section{Introduction}

The \emph{motion-planning-through-gadgets framework}, introduced in \cite{gadgets} and further developed in \cite{gadgets2}, captures a broad range of combinatorial motion-planning problems. It also serves as a powerful tool for proving hardness of games and puzzles that involve an agent moving in and interacting with an environment where the goal is to reach a specified location. Prior work \cite{gadgets2} fully characterizes the complexity of 1-player motion planning with two natural classes of gadgets: \emph{DAG $k$-tunnel} gadgets, which naturally lead to bounded games, and \emph{reversible deterministic $k$-tunnel} gadgets, which naturally lead to unbounded games. It also extends the gadget model to 2-player and team imperfect information variants, and provides full or partial characterizations for the bounded and unbounded versions of each.

All of the prior work considers \emph{reachability}, where the decision problem is whether the agent can reach the target location. In this paper, we begin extending the gadget model to victory conditions other than reaching a target location. In particular we examine the complexity of reconfiguring a system of gadgets and of visiting every single gadget.

We consider the \emph{universal traversal} problem of whether the agent can visit every gadget. In Section~\ref{sec:visit every}, we characterize the complexity of this problem for three classes of $k$-tunnel gadgets: DAG gadgets, one-state gadgets, and reversible deterministic gadgets. Of particular note is that universal traversal can be harder than reachability for the same gadget. In particular, there are DAG $k$-tunnel gadgets for which reachability is in P but universal traversal is NP-complete. Additionally, reachability for one-state gadgets is always in P, but universal traversal can be NP-complete.

In Section~\ref{sec:reconfiguration} we consider the \emph{reconfiguration} problem of whether the agent can cause the gadgets to reach a target configuration. We show a gadget with non-interacting tunnels whose reconfiguration problem is PSPACE-complete, but whose reachability problem is in P. We also show that for reversible, deterministic gadgets PSPACE-completeness of the reachability problem implies PSPACE-completeness of the reconfiguration problem. In contrast we exhibit a gadget for which the reconfiguration problem becomes easier, contained in P whereas reachability is NP-complete. We also extend the results in \cite{gadgets2} to show a larger class of gadgets which are in NP and describe a sub-class of those which are in P. The gadgets framework has already been used to prove complexity results about reconfiguration problems related to swarm \cite{balanza2019full} and modular robotics \cite{akitaya2021characterizing}, so understanding reconfiguration in the gadgets model may provide an easier and more powerful base for such applications.

\section{Gadget Model}

We now define the gadget model of motion planning, introduced in \cite{gadgets}.

A \emph{gadget} consists of a finite number of
\emph{locations} (entrances/exits) and a finite number of \emph{states}.
Each state $S$ of the gadget defines a labeled directed graph on the locations,
where a directed edge $(a,b)$ with label $S^\prime$ means that an agent
can enter the gadget at location $a$ and exit at location $b$, changing the state of the gadget from $S$ to $S^\prime$. Each of these arcs is called a \emph{transition}. Sometimes we will discuss a \emph{traversal} from some location $a$ to location $b$ which refers to any possible transition from $a$ to $b$ in state $s$. Different states in a gadget can have different transitions while having the same traversability.
Equivalently, a gadget is specified by its \emph{transition graph},
a directed graph whose vertices are state/location pairs,
where a directed edge from $(S,a)$ to $(S',b)$ represents that the agent
can traverse the gadget from $a$ to $b$ if it is in state~$S$,
and that such traversal will change the gadget's state to~$S^\prime$. 
Gadgets are \emph{local} in the sense that traversing a gadget does
not change the state of any other gadgets. An example can be seen in Figure~\ref{fig:L2T}

\begin{figure}
    \centering	\includegraphics[scale=1]{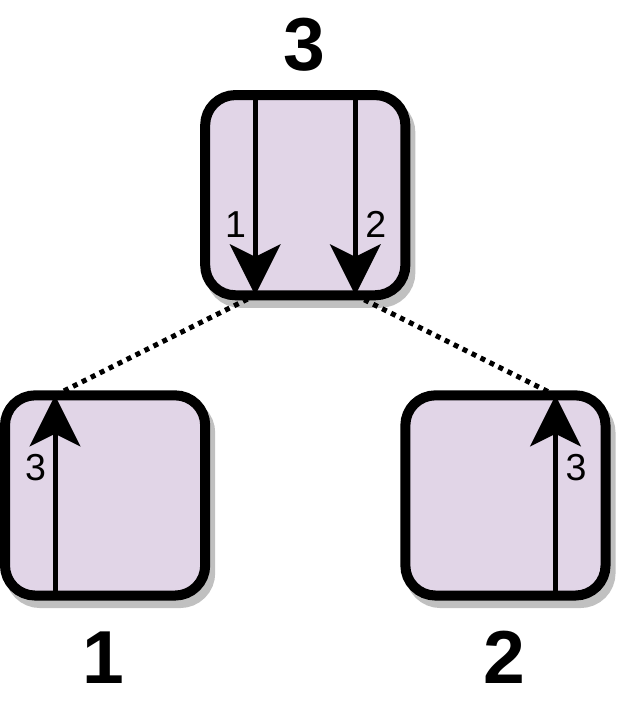}
	\caption{A diagram describing the locking 2-toggle gadget. Each box represents the gadget in a different state, in this case labeled with the numbers $1,2,3$. Arrows represent transitions in the gadget and are labeled with the states to which those transition take the gadget. In the top state 3, the agent can traverse either tunnel going down, which blocks off the other tunnel until the agent reverses the initial traversal.}
	\label{fig:L2T}
\end{figure}

A \emph{system of gadgets} consists of gadgets, the initial state of each gadget, and
a \emph{connection graph} on the gadgets' locations.
If two locations $a$ and $b$ of two gadgets (possibly the same gadget) are connected
by a path in the connection graph, then an agent can traverse freely between
$a$ and~$b$ (outside the gadgets). The \emph{configuration} of a system of gadgets is the system along with a state for each of the gadgets in the system.
(Equivalently, we can think of locations $a$ and $b$ as being identified,
effectively contracting connected components of the connection graph.)
These are all the ways that the agent can move: exterior to gadgets using
the connection graph, and traversing gadgets according to their current states.

Previous work has focused on the reachability problem \cite{gadgets,gadgets2}:

\begin{definition}
	For a set of gadgets gadget $F$, \emph{reachability for $F$} is the following decision problem. Given a system of gadgets consisting of copies of gadgets in $F$, and a starting location and a win location in that system of gadgets, is there a path the agent can take from the starting location to the win location?
\end{definition}

A 2-player and a team multiplayer version of this problem was also defined and investigated\cite{gadgets2}.

We will consider several specific classes of gadgets.

\begin{definition}
	A \emph{$k$-tunnel} gadget has $2k$ locations, which are partitioned into $k$ pairs called \emph{tunnels}, such that every transition is between two locations in the same tunnel.
\end{definition}

\begin{definition}
	The \emph{state-transition graph} of a gadget is the directed graph which has a vertex for each state, and an edge $S\to S^\prime$ for each transition from state $S$ to $S^\prime$. A \emph{DAG} gadget is a gadget whose state-transition graph is acyclic.
\end{definition}

DAG gadgets naturally lead to bounded problems, since they can be traversed a bounded number of times. The complexity of the reachability problem for DAG $k$-tunnel gadgets, as well as the 2-player and team games, is characterized in \cite{gadgets2}.

\begin{definition}
	A gadget is \emph{deterministic} if every traversal can put it in only one state and every location has at most
	1 traversal from it.
	More precisely, its transition graph has maximum out-degree 1.
\end{definition}

\begin{definition}
	A gadget is \emph{reversible} if every transition can be reversed. More precisely, its transition graph is undirected.
\end{definition}

\begin{definition}
	A gadget has a \emph{distant opening} if a transition in some state across a tunnel which opens a different tunnel. A tunnel is \emph{opened} if a transition has taken it from a state where the tunnel did not have traversability in some direction to a state where it is now traversable.
\end{definition}

Reversible deterministic gadgets are gadgets whose transition graphs are partial matchings, and they naturally lead to unbounded problems. Prior work \cite{gadgets2} characterizes the complexity of reachability for reversible deterministic $k$-tunnel gadgets and partially characterizes the complexity of the 2-player and team games.

In Section~\ref{sec:stateless visit}, we also consider \emph{one-state}, $k$-tunnel gadgets. A transition in a gadget with only one state cannot change the state, so the legal traversals never changes.

\section{Universal Traversal}\label{sec:visit every}

In this section, we consider the question of whether an agent in a system of gadgets can make a traversal across every gadget, called the \emph{universal traversal} problem.

\begin{definition}
	For a gadget $G$, \emph{universal traversal for $G$} is the following decision problem. Given a system of gadgets consisting of copies of $G$ and a starting location, is there a path the agent can take from the starting location which makes at least one traversal in every gadget?
\end{definition}

We provide a full characterization for the complexity of this problem for three classes of gadgets. In Section~\ref{sec:bounded visit}, we characterize DAG $k$-tunnel gadgets. Universal traversal is NP-hard for some DAG gadgets where reachability is in P. This is somewhat similar to the distinction between finding paths and finding Hamiltonian paths. In Section~\ref{sec:stateless visit}, we further emphasize this difference by characterizing one-state $k$-tunnel gadgets. Reachability is always in NL for one-state gadgets, but we find that universal traversal is often NP-complete. Finally, Section~\ref{sec:unbounded visit} considers the unbounded case by characterizing universal traversal for reversible deterministic $k$-tunnel gadgets. In this case, the dichotomy is the same as for reachability.

\begin{lemma}\label{lem:visiting in pspace}
	Universal traversal for any gadget is in PSPACE.
\end{lemma}

\begin{proof}
	We can easily solve the universal traversal problem in nondeterministic polynomial space: repeatedly guess the next traversal, keeping track of which gadgets have been used, and accept once they all have been. By Savitch's theorem \cite{SAVITCH} universal traversal is in PSPACE.
\end{proof}

\subsection{DAG Gadgets}\label{sec:bounded visit}

In this subsection, we consider universal traversal for $k$-tunnel DAG gadgets. We find that this problem is NP-hard for any DAG gadget which has and actually uses at least 2 tunnels, in the sense defined below. For some simple 1-tunnel DAG gadgets, universal traversal is analogous to finding Eulerian paths and is thus in P; however, more complex 1-tunnel DAG gadgets can not easily be converted to an Eulerian path problem. For example the 1-toggle which switches direction after each transition or a gadget which can be traversed at most twice. We leave the case of 1-tunnel DAG gadgets open.

\begin{problem}
	Is universal traversal with any 1-tunnel DAG gadget in P? Are there 1-tunnel DAG gadgets for which universal traversal is NP-complete?
\end{problem}

Some $k$-tunnel DAG gadgets with $k>1$ act like 1-tunnel gadgets in that it is never possible to make use of multiple tunnels. A simple example is shown in Figure~\ref{fig:not 2-tunnel}. We formalize this notion in the following definition.

\begin{figure}
	\centering
	\includegraphics[width=.4\linewidth]{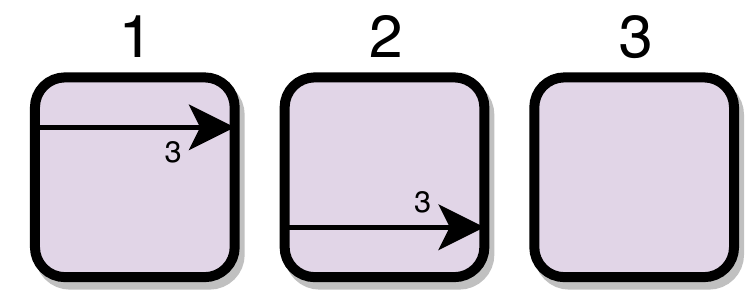}
	\caption{A 2-tunnel DAG gadget which is not true 2-tunnel.}
	\label{fig:not 2-tunnel}
\end{figure}

\begin{definition}
	A state of a $k$-tunnel gadget is \emph{true 2-tunnel} if there are at least two tunnels which are traversable in any state reachable from that state. A gadget is \emph{true 2-tunnel} if it has a true 2-tunnel state.
\end{definition}

Note that a $k$-tunnel gadget does not need multiple tunnels traversable in the same state to be true 2-tunnel: perhaps traversing the single traversable tunnel opens another tunnel. To justify this definition, we prove the following result.

\begin{theorem}
	Let $G$ be a $k$-tunnel which is not true 2-tunnel. Then there is a $1$-tunnel gadget $G^\prime$ and a bijection between states of $G$ to states of $G^\prime$ such that replacing each copy of $G$ in a system of gadgets with a copy of $G^\prime$ in the corresponding state gives an equivalent system of gadgets.
\end{theorem}

We leave ambiguous exactly what `equivalent' means, since what counts as equivalent should be different for different victory conditions. For example, any gadget simulation is sufficient to show hardness for reachability, but this may not suffice for universal traversal because traversing the simulation does not necessarily involve traversing every gadget inside it. In the case at hand, the systems of gadgets are equivalent in that the answers to the reachability and universal traversal problems are all the same, and the proof should extend to other victory conditions one may consider, though not necessarily all of them.

\begin{proof}
	To construct $G^\prime$, we simply collapse the $2k$ locations in $G$ to $2$ locations by merging all of the tunnels. From any state in $G$, there is only one tunnel which can ever be traversable. Thus ignoring all of the other tunnels yields the same gadget. If there are different states in $G$ which have different traversable tunnels, we can move them to the same tunnel since these states are never reachable from each other.
\end{proof}

We will use the fact that every nontrivial DAG gadget simulates either a directed or an undirected single-use path, since we can take a final nontrivial state of the gadget \cite{gadgets2}.

The rest of this subsection is devoted to proving NP-completeness for universal traversal for true 2-tunnel DAG gadgets.

\begin{theorem}\label{thm:visit dag}
	Universal traversal for any true 2-tunnel DAG gadget is NP-complete.
\end{theorem}

To prove Theorem~\ref{thm:visit dag}, we will focus on the \emph{final} true 2-tunnel state of a DAG gadget, and only use the two tunnels which make this state true 2-tunnel. Such a state exists because the state-graph is a DAG. After making a traversal in this state, any resulting state is not true 2-tunnel, so only one of the two tunnels can be traversed in the future. Note that if the gadget is nondeterministic, the agent may be able to choose which of the two tunnels this is. We will consider several cases for the form of the last true 2-tunnel state, and show NP-hardness for each one.

The first case we consider is when the final true 2-tunnel state has a distant opening.

\begin{lemma}\label{lem:visit dag open}
	Let $G$ be a true 2-tunnel gadget and let $S$ be a final true 2-tunnel state of $G$. If any transition from $S$ across one tunnel opens a traversal across another tunnel, then universal traversal for $G$ is NP-hard.
\end{lemma}

\begin{proof}
	We will only use the two tunnels involved in the opening transition from $S$ to $S^\prime$. Suppose traversing the top tunnel from left to right allows the agent to open the left-to-right traversal on the bottom tunnel. Then state $S$ has one of the two forms shown in Figure~\ref{fig:visit dag open}, depending on whether the bottom tunnel can be traversed right to left in $S$. In either case, the top tunnel may or may not be traversable from right to left in $S$. Since $S$ is a final true 2-tunnel state, only the bottom tunnel is traversable in $S^\prime$.

	\begin{figure}
		\centering
		\begin{subfigure}{.3\linewidth}
			\centering
			\includegraphics[width=\linewidth]{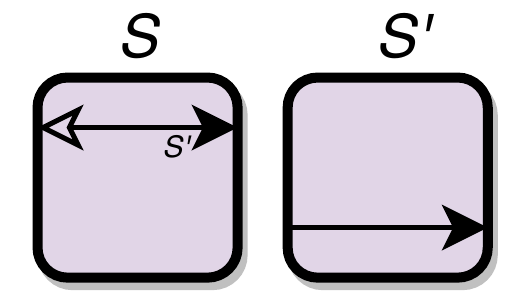}
			\caption{}
			\label{subfig:dag open no}
		\end{subfigure}
		\hfil
		\begin{subfigure}{.3\linewidth}
			\centering
			\includegraphics[width=\linewidth]{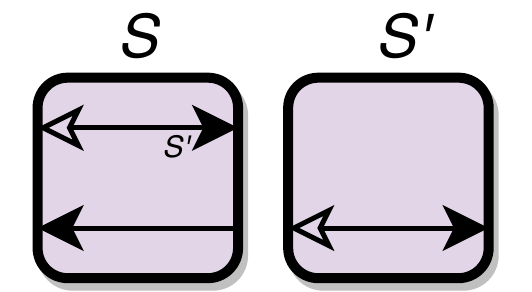}
			\caption{}
			\label{subfig:dag open yes}
		\end{subfigure}
		\caption{Two cases for the form of the gadget in Lemma~\ref{lem:visit dag open}, assuming traversing the top tunnel to the right opens the bottom tunnel to the right. In (a) the bottom tunnel is not traversable to the left in state $S$ and in (b) it is. Unfilled arrows are traversals that may or may not exist depending on the gadget. Unlabled transitions may be to arbitrary states not specified here.}
		\label{fig:visit dag open}
	\end{figure}

	To show NP-hardness of universal traversal, we reduce from the reachability problem for the same gadget. Since the gadget has a distant opening, reachability is NP-complete \cite{gadgets2}. We modify the system of gadgets in an instance of the reachability problem by adding a construction to each gadget which allows the agent to go back and make a traversal in it after reaching the win location. If the agent can reach the win location, it can then use any gadgets it did not already use, and if it cannot reach the win location, it cannot use the gadgets involved in this construction.

	\begin{figure}
		\centering
		\includegraphics[width=.4\linewidth]{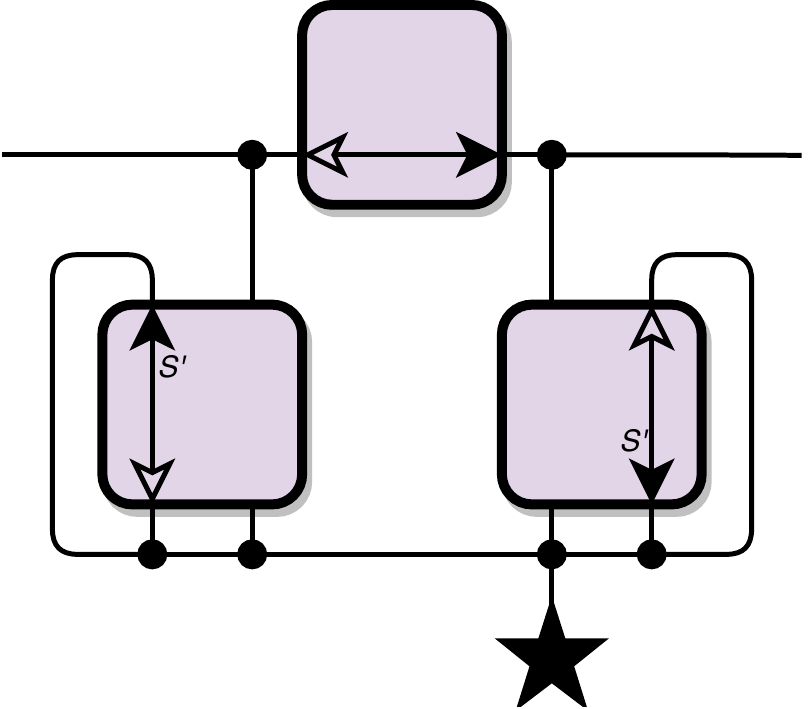}
		\caption{The construction to allow the agent to use a gadget after reaching the win location, when the bottom tunnel is not traversable in state $S$ (the case of Figure~\ref{subfig:dag open no}). The star denotes the goal location.}
		\label{fig:dag open no attachment}
	\end{figure}

	The construction is slightly different depending on whether the bottom tunnel can be traversed from right to left in state $S$. We use the construction in either Figure~\ref{fig:dag open no attachment} or Figure~\ref{fig:dag open yes attachment}. In either case, the agent cannot use the newly added gadgets until it first reaches the win location. Once it reaches the win location, it can open tunnels in the added gadgets, traverse the (top) gadget the construction is attached to, and return. If the agent already used the gadget this is attached to, it can instead use a traversal in each added gadget without visiting that gadget. So it is possible to make a traversal in every gadget if and only if the original reachability problem is solvable.
\end{proof}

	\begin{figure}
		\centering
		\includegraphics[width=.4\linewidth]{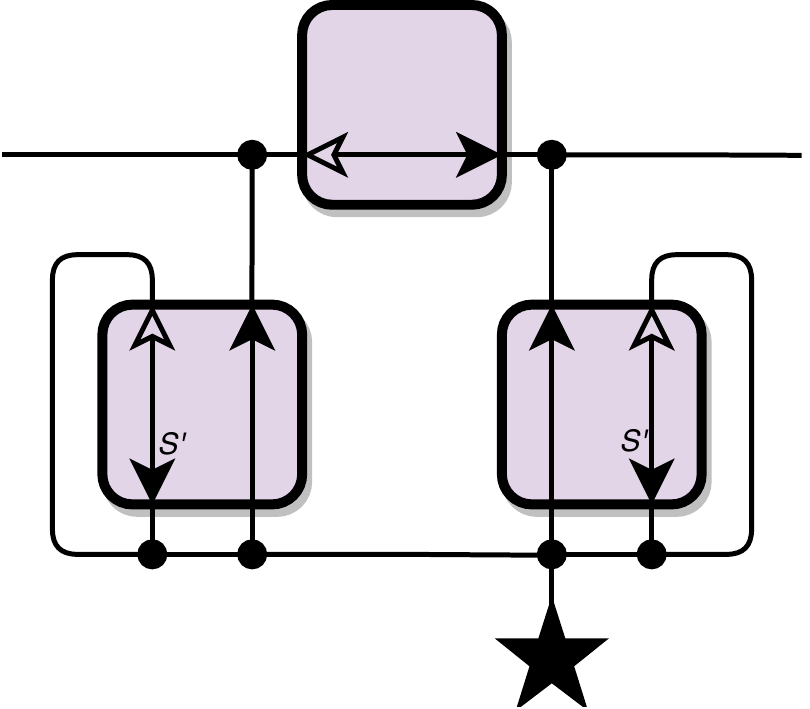}
		\caption{The construction to allow the agent to use a gadget after reaching the win location, when the bottom tunnel is traversable from right to left in state $S$ (the case of Figure~\ref{subfig:dag open yes}).}
		\label{fig:dag open yes attachment}
	\end{figure}

Now we will assume the final true 2-tunnel state has no distant opening. If only one tunnel is traversable in this state, then it cannot be true 2-tunnel because the other tunnel will never become traversable. So both tunnels are traversable, and after making any traversal, only one tunnel will ever be traversable. With no distant opening, we first consider the case where at least one of the tunnels is directed in the final true 2-tunnel state.

\begin{lemma}\label{lem:visit dag no open dir}
	Let $G$ be a true 2-tunnel gadget and let $S$ be a final true 2-tunnel state of $G$. Suppose no transition from $S$ across one tunnel opens a traversal across the other tunnel. If, in $S$, some tunnel can be traversed in one direction but not in the other, then universal traversal for $G$ is NP-hard.
\end{lemma}

\begin{proof}
	A directed tunnel with a single-use path on each side is a single-use directed path; since $G$ has a directed tunnel in state $S$, it simulates a single-use directed path.

	We reduce from finding a Hamiltonian path in a directed 3-regular graph with specified start and end vertices $s$ and $t$. Each vertex of the graph other than $s$ and $t$ becomes one of the vertex gadgets in Figure~\ref{fig:visit dag dir}, depending on its in-degree. We replace $s$ with the right half of the appropriate vertex gadget and $t$ with the left half. The agent begins at $s$.

	\begin{figure}
		\centering
		\includegraphics{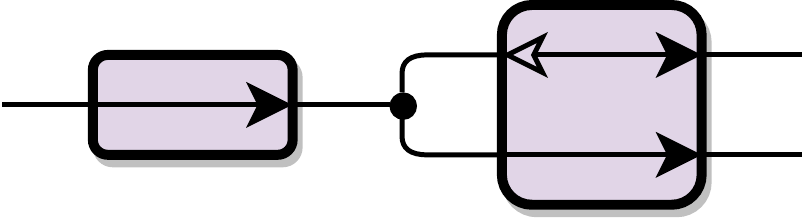}

		\vspace{.5cm}

		\includegraphics{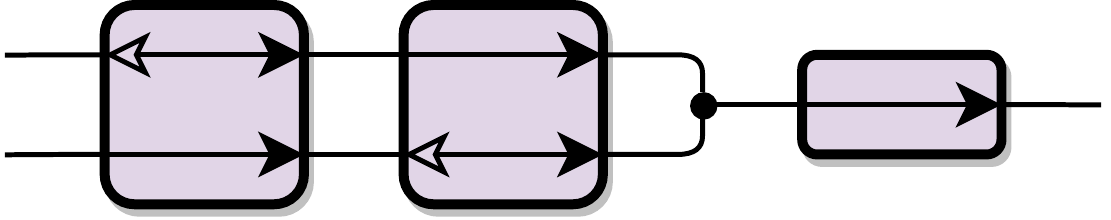}
		\caption{Vertex gadgets for Lemma~\ref{lem:visit dag no open dir}. The first construction is for vertices with in-degree 1 and out-degree 2, and the second construction is for vertices with in-degree 2 and out-degree 1. Each construction contains one or two copies of $G$ in state $S$ and one single-use directed path. By assumption, state $S$ contains two traversable tunnels, at least one of which is directed. If both tunnels are directed, we only need one of the gadgets in state $S$ for the in-degree 2 vertex gadget.}
		\label{fig:visit dag dir}
	\end{figure}

	If there is a Hamiltonian path, the agent can follow it and thereby make a traversal in every gadget by going through every vertex gadget. Suppose the universal traversal problem is solvable. The agent must use the single-use directed path in each vertex gadget, and thus must go through every vertex. Suppose it enters a vertex gadget with in-degree 2 along the top path, and reaches the vertex in the center. Since, by assumption, making a traversal across a tunnel in $S$ cannot open a traversal on the other tunnel, the bottom tunnel of the left gadget still is not traversable to the left, so the agent cannot exit on the bottom path. Similarly it cannot enter on the bottom path and exit on the top path. Next, the agent cannot enter a vertex gadget with in-degree 1 for the first time on either path on the right, since this would require exiting another vertex gadget to the left on a path it has not used before. If the agent exits a vertex gadget with in-degree 1 on the left, that copy of $G$ is now not true 2-tunnel, so the agent cannot later enter and exit on different paths on the right. Summarizing, the agent always enters vertex gadgets on the left and exits on the right, and it cannot use all three entrances or exits of a vertex gadget. Thus the agent's path corresponds to a path in the graph, and since it must use each single-use directed path this path is Hamiltonian.
\end{proof}

The remaining case is when, in the final true 2-tunnel state, there is no distant opening and all tunnels are undirected.  We branch into two cases one last time, based on whether traversing one tunnel requires closing the other tunnel.

\begin{lemma}\label{lem:visit dag no open undir close}
	Let $G$ be a true 2-tunnel gadget and let $S$ be a final true 2-tunnel state of $G$. Suppose there are two tunnels $a$ and $b$ which can both be traversed in both directions in $S$. Furthermore, suppose that every transition from $S$ across $a$ from left to right goes to a state in which $b$ cannot be traversed from right to left. Then universal traversal for $G$ is NP-hard.
\end{lemma}

\begin{proof}
	The form of state $S$ is shown in Figure~\ref{fig:dag undir}. We reduce from finding a Hamiltonian path in a directed 3-regular graph with specified start and end vertices $s$ and $t$, similarly to Lemma~\ref{lem:visit dag no open dir}. We replace each vertex other than $s$ and $t$ with the appropriate vertex gadget in Figure~\ref{fig:visit dag undir}, and replace $s$ and $t$ with the appropriate half of one of these vertex gadgets. If there is a Hamiltonian path, the agent can follow it to make a traversal in every gadget.

	\begin{figure}
		\centering
		\includegraphics[width=.3\linewidth]{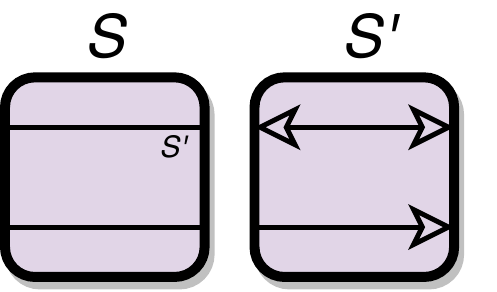}
		\caption{The form of the gadget in Lemma~\ref{lem:visit dag no open undir close}. Every transition from state $S$ across the top tunnel to the right closes the right-to-left traversal on the bottom tunnel.}
		\label{fig:dag undir}
	\end{figure}

	Suppose the agent can make a traversal in every gadget; we consider how it moves through each vertex gadget. It must go across the single-use path in each vertex gadget. Suppose the agent enters a vertex gadget with in-degree 1 on the single-use path. It must exit on a path on the right. If it returns to the vertex gadget along a path on the right, it cannot leave on the other path since at this point that copy of $G$ is not true 2-tunnel; so the agent cannot accomplish anything by returning to the vertex gadget. Now suppose it enters a vertex gadget with in-degree 2 along a path on the left. Because every transition from $S$ crossing $a$ to the right closes the right-to-left traversal of $b$, the agent cannot exit the vertex gadget across the other left path. It can return to where it was, or exit across the single-use path.

	\begin{figure}
		\centering
		\includegraphics{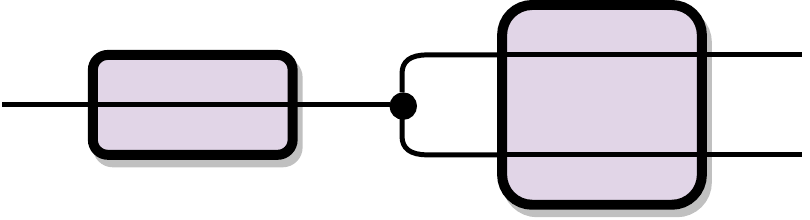}

		\vspace{.5cm}

		\includegraphics{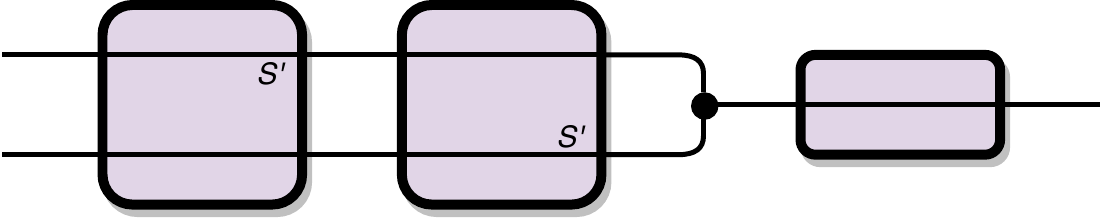}
		\caption{Vertex gadgets for Lemma~\ref{lem:visit dag no open undir close}. The first construction is for vertices with in-degree 1 and out-degree 2, and the second construction is for vertices with in-degree 2 and out-degree 1. Each construction contains one or two copies of $G$ in state $S$ and one single-use path.}
		\label{fig:visit dag undir}
	\end{figure}

	In particular, by induction the agent must always enter a vertex gadget on one of the in-edges and exit on an out-edge. It cannot use more than two edges on each vertex gadget, and must use the single-use path. So its path corresponds to a Hamiltonian path from $s$ to $t$.
\end{proof}

\begin{lemma}\label{lem:visit dag no open undir no close}
	Let $G$ be a true 2-tunnel gadget and let $S$ be a final true 2-tunnel state of $G$. Suppose there are two tunnels $a$ and $b$ which can both be traversed in both directions in $S$. Furthermore, suppose that both traversals from state $S$ across $a$ can leave either direction across $b$ traversable, and vice-versa. The universal traversal for $G$ is NP-hard.
\end{lemma}

\begin{proof}
	We reduce from finding a Hamiltonian path in an undirected 3-regular graph with specified start and end vertices $s$ and $t$, assuming $s$ and $t$ have degree $1$ (so the graph is not quite 3-regular). We will only use state $S$ and the tunnels $a$ and $b$. Each vertex of the graph other than $s$ and $t$ is replaced with the construction in Figure~\ref{fig:spiraltonian}, where each of the nine gadgets is in state $S$ and the tunnels involved are $a$ and $b$. The start location is at $s$. There is a single-use path leading to $t$; this forces the agent to end by entering $t$.

	\begin{figure}
		\centering
		\includegraphics[width=.7\linewidth]{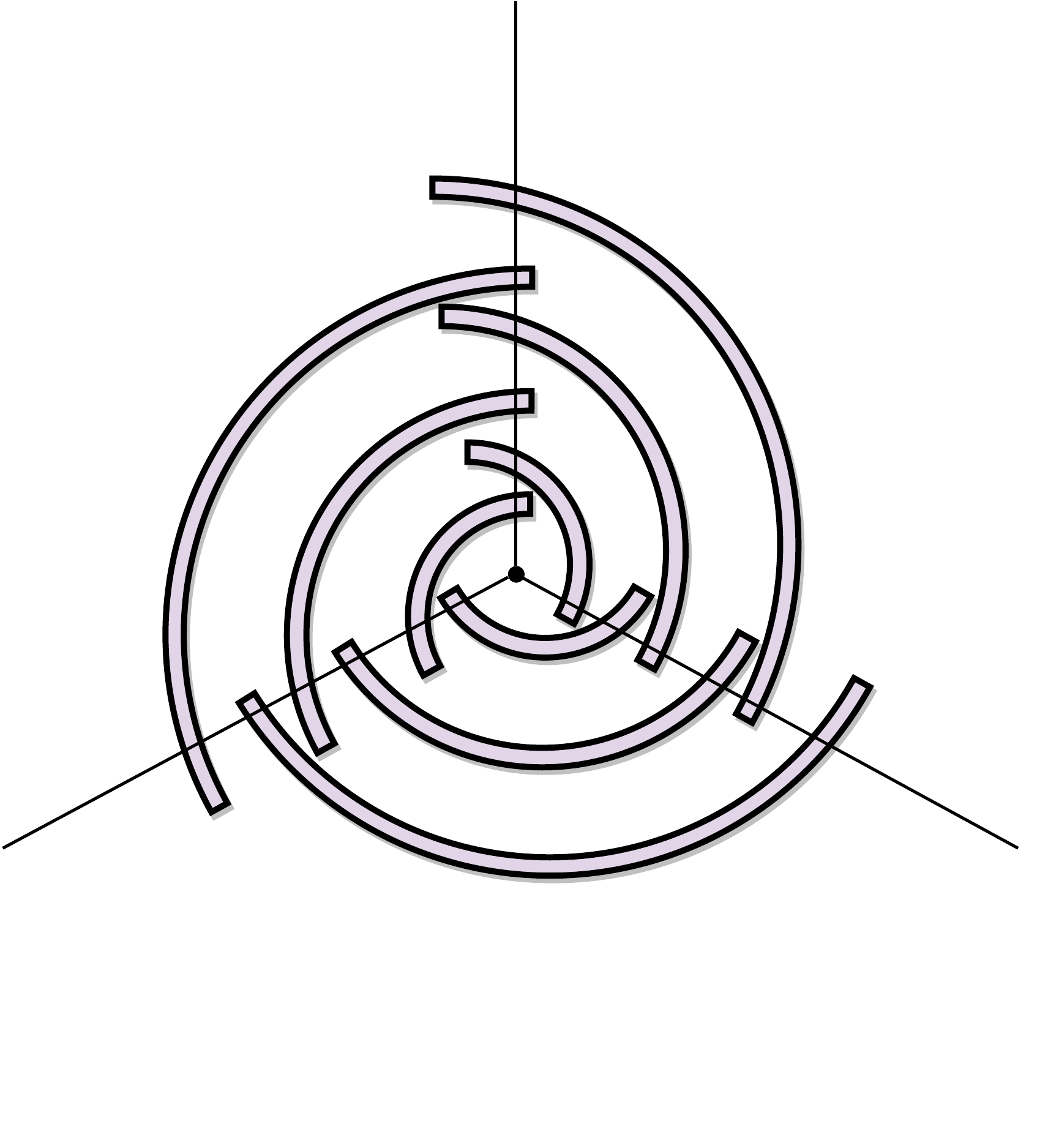}
		\caption{A vertex gadget for Lemma~\ref{lem:visit dag no open undir no close}. The individual gadgets are true 2-tunnel gadgets and are shown elongated in this diagram.}
		\label{fig:spiraltonian}
	\end{figure}

	Suppose there is a Hamiltonian path. Then the agent will follow the path, and thereby traverse every gadget. At each vertex, by assumption all four traversals across each gadget are currently open. As the agent moves towards the center of the vertex gadget, it will choose transitions so that the path out another line stays open; this is possible by assumption. So it is able to follow the Hamiltonian path. Traversing a vertex in this way goes through every gadget in it, so since the path is Hamiltonian the agent uses all of these gadgets. Since the path ends at $t$, it also uses the single-use path to $t$.

	Conversely, suppose the agent can make a traversal in every gadget. We consider ways it can go through each vertex gadget. Suppose it enters and leaves on different paths; there are three gadgets it uses both tunnels on. After the first traversal through one of these gadgets, the other tunnel must be open, so the first tunnel used is permanently closed since (by assumption) $S$ is a final true 2-tunnel state. In particular, the agent cannot use the path it entered on again later. It also cannot later enter on one of the other two paths and leave on the other one, since that would require using both tunnels on the gadgets which intersect both paths after the gadget was already used when the agent exited earlier. So the agent cannot go through the vertex gadget (meaning it enters and leaves on different paths) more than once.

	In order to use one of the innermost gadgets in the vertex gadget, the agent must reach the vertex in the center. Suppose that after doing so for the first time, it exits along the same path it entered. Then, again since $S$ is a final true 2-tunnel state, the tunnel not on that path of each gadget which has a tunnel on that path is permanently closed. So the middle gadget opposite the path used is entirely cut off from the rest of the graph, and can never be used. Similarly the agent cannot have previously traversed that gadget by partially entering the vertex gadget without reaching the vertex in the center. So it must go through the vertex gadget at least once.

	Since the agent must go through each vertex gadget exactly once, starts at $s$, and must end at $t$ because of the single-use path which it must also use, its path is a Hamiltonian path from $s$ to $t$ in the original graph. Note that it could go through a vertex gadget, then later enter on the other path, but all it can then do is exit on that path so this does not accomplish anything.

	In summary, the universal traversal problem is solvable if and only if a Hamiltonian path exists. Since this variant of Hamiltonian path is NP-hard, so is the universal traversal problem.
\end{proof}

These cases together cover every true 2-tunnel DAG gadget, so we can now prove Theorem~\ref{thm:visit dag}, that universal traversal for any true 2-tunnel DAG gadget is NP-complete.

\begin{proof}
	Since the gadget is a DAG, the agent can make a bounded number of traversals in each copy of the gadget. So the solution path has polynomial length, and thus the problem is in NP.

	For NP-hardness, we consider a final true 2-tunnel state $S$ and use one of the preceding lemmas. If a transition from $S$ across some tunnel opens a traversal across a different tunnel, NP-hardness follows from Lemma~\ref{lem:visit dag open}. Otherwise, if $S$ contains a directed tunnel, we have NP-hardness from Lemma~\ref{fig:visit dag dir}. Otherwise, all tunnels traversable in $S$ are traversable in both directions, and no transition from $S$ opens a tunnel. If traversing some tunnel in some direction from $S$ forces the agent to close some traversal across another tunnel, NP-hardness follows from Lemma~\ref{lem:visit dag no open undir close}. Finally, if there is no traversal with that property, Lemma~\ref{lem:visit dag no open undir no close} gives NP-hardness. Together these lemmas cover all true 2-tunnel DAG gadgets.
\end{proof}

For a DAG gadget, universal traversal and reachability can have different complexity. Reachability is NP-hard if and only if the gadget has a distant opening or a forced distant closing \cite{gadgets2}. Each of these properties implies that the gadget is true 2-tunnel, so universal traversal is NP-hard whenever reachability is. However, sometimes reachability is in P while universal traversal is NP-hard. For example, for the gadget shown in Figure~\ref{fig:visiting harder}, NP-hardness of universal traversal is given by Lemma~\ref{lem:visit dag no open dir}, whereas reachability is in NL since there are not interacting tunnels.

\begin{figure}
	\centering
	\includegraphics[width=.7\linewidth]{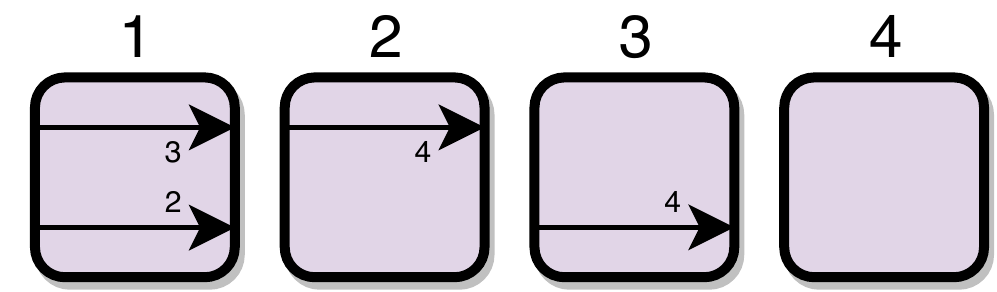}
	\caption{A DAG gadget for which reachability is in NL but universal traversal is NP-hard. Crossing either directed tunnel closes that tunnel without affecting the other tunnel.}
	\label{fig:visiting harder}
\end{figure}

More generally, the gadgets considered in Lemmas~\ref{lem:visit dag open} and \ref{lem:visit dag no open undir close} all have NP-hard reachability as well as universal traversal, but those considered in Lemmas~\ref{lem:visit dag no open dir} and \ref{lem:visit dag no open undir no close} do not necessarily have NP-hard reachability despite universal traversal being NP-hard.

The proofs of Lemmas~\ref{lem:visit dag no open dir}, \ref{lem:visit dag no open undir close}, and \ref{lem:visit dag no open undir no close} can be considered as reductions from finding Hamiltonian paths in planar graphs, which shows the universal traversal problem NP-hard even when restricted to planar systems of gadgets. This leaves open the question of whether this is also true for the gadgets considered in Lemma~\ref{lem:visit dag open}.

\begin{problem}
 Is universal traversal restricted to planar systems of gadgets NP-hard for all true 2-tunnel DAG gadgets?
\end{problem}

\subsection{One-State Gadgets}\label{sec:stateless visit}

In this subsection, we consider universal traversal for $k$-tunnel gadgets with only one state. The reachability problem is clearly in NL for such gadgets, but we will see that universal traversal is often NP-complete. This is a broader analogue of the distinction between reachability and universal traversal that we saw for DAG gadgets in the previous section.

A one-state $k$-tunnel gadget consists of directed and undirected tunnels, and is determined by the number of each type; we assume there is no untraversable tunnel since such a tunnel can be removed without affecting the problem. We fully characterize the complexity of universal traversal for such gadgets.

\begin{theorem}\label{thm:stateless visit}
	Let $G$ be a one-state $k$-tunnel gadget. If $G$ has no directed tunnels, then universal traversal for $G$ is in L. Otherwise, if $k\leq2$ universal traversal for $G$ is NL-complete and if $k\geq3$ universal traversal for $G$ is NP-complete.
\end{theorem}

We will prove each portion of Theorem~\ref{thm:stateless visit} in a separate lemma. 

\begin{lemma}\label{lem:stateless np}
	Universal traversal with any one-state gadget is in NP.
\end{lemma}

\begin{proof}
	If there is a way to use every gadget, this can be done in a number of traversals at most quadratic in the number of tunnels: list the gadgets in an order they can all be visited, and take the shortest path between each pair. The number of gadgets and each such shortest path has length at most the total number of tunnels. So the full solution path can be described in polynomial space. We use this as a certificate; clearly we can check in polynomial time whether a potential solution works.
\end{proof}

\begin{lemma}\label{lem:stateless undir}
	Universal traversal with any one-state $k$-tunnel gadget with no directed tunnel is in L.
\end{lemma}

\begin{proof}
	We solve the universal traversal problem as follows. Iterate over each gadget. For each one, iterate over its locations. For each location, we check whether there is a path from the start location to that location; this is reachability in an undirected graph which can be solved in logarithmic space \cite{connectivity}. If there is a path, we move on to the next gadget. If there is no path, we move on to the next location on the gadget, unless this was the last location, in which case we reject. After finishing all gadgets, we accept.

	This algorithm can clearly run in logarithmic space. Since all tunnels are undirected, the agent can visit each gadget in turn and return to the start location after each one. So the agent can use every gadget exactly when there is a path to every gadget from the start location, which is what the algorithm checks.
\end{proof}

\begin{lemma}\label{lem:stateless nl-hard}
	Universal traversal with any one-state $k$-tunnel gadget with a directed tunnel is NL-hard.
\end{lemma}

\begin{proof}
	We reduce from $s$-$t$ connectivity in directed graphs, which is NL-complete \cite{connectivity}. We will use only one directed tunnel in each gadget.

	Given a directed graph with vertices $s$ and $t$, we first add edges $t\to v$ and $v\to s$ for each vertex $v$. Then we replace each edge with a directed tunnel in a gadget. This can clearly be done in logarithmic space. 

	If there is no path from $s$ to $t$, then the agent can never traverse the tunnel $t\to s$. If there is such a path, the agent can go to $t$, go to the entrance of a tunnel, go through the tunnel, and return to $s$. By doing this for each edge in the graph, the agent can make a traversal in every gadget. So the universal traversal problem is solvable exactly when there is a path from $s$ to $t$.
\end{proof}

\begin{lemma}\label{lem:stateless nl}
	Universal traversal with any one-state $k$-tunnel gadget is in NL if $k\leq2$. 
\end{lemma}

\begin{proof}
	We provide an algorithm which runs in nondeterministic logarithmic spaces with an oracle for reachability in directed graphs. This shows that the universal traversal problem is in NL$^{\text{NL}}$. We will then explain how the algorithm can be adapted to run in NL. The algorithm first uses the oracle to convert the problem to an instance of 2SAT. It then solves this instance, since 2SAT is in NL.

	The 2SAT formula has a variable for each tunnel in the system of gadgets; a satisfying assignment will provide a set of tunnels we can traverse to solve the universal traversal problem. For each gadget with tunnels $x_1$ and $x_2$, we have a clause $x_1\vee x_2$ (if the gadget has only one tunnel, $x_1=x_2$). For each pair of distinct tunnels $x$ and $y$, we query the reachability oracle to determine whether there is a path from the exit of $x$ to the entrance of $y$ or from the exit of $y$ to the entrance of $x$ (if $x$ or $y$ is undirected, we can use either location as the entrance or exit). If there is no path in either direction, we have a clause $\neg x\vee\neg y$.

	We prove that this algorithm works, and then adapt it to an $\text{L}^{\text{NL}}$ algorithm which is known to equal NL \cite{nl=conl}.

	\begin{lemma}
		The 2SAT formula above is satisfiable if and only if the universal traversal problem has a solution.
	\end{lemma}

	\begin{proof}
	First suppose the universal traversal problem is solvable, and consider the assignment which contains the tunnels which are used in the solution. Since the solution must use a tunnel in every gadget, each clause $x_1\vee x_2$ is satisfied. If the solution uses both tunnels $x$ and $y$, there must be a path in some direction between $x$ and $y$, namely the path the agent takes between the two tunnels. For each clause $\neg x\vee\neg y$ in the formula, there is no such path, so the solution does not use both tunnels $x$ and $y$, so the clause is satisfied.

	Now suppose the 2SAT formula is satisfiable, and consider the set $T$ of tunnels corresponding to true variables in a satisfying assignment. Because of the clauses $x_1\vee x_2$, $T$ must contain a tunnel in each gadget. We define a relation $\rightarrow$ on $T$ where $x\rightarrow y$ if there is a path from the exit of $x$ to the entrance of $y$. As suggested by the notation, this relation is transitive: if $x\rightarrow y\rightarrow z$, there is a path from the exit of $x$ to the entrance of $y$, across $y$, and then to the entrance of $z$, so $x\rightarrow z$. Since each clause $\neg x\vee\neg y$ is satisfied, for any distinct $x,y\in S$ we have $x \rightarrow y$ or $y\rightarrow x$. That is, $\rightarrow$ is a strict total pre-order.

	Then there must be a (strict) total order $\prec$ on $T$ such that $x\prec y\implies x\rightarrow y$: define another relation $\sim$ where $x\sim y$ if $x=y$ or both $x\rightarrow y$ and $y\rightarrow x$. Then $\sim$ is clearly an equivalence relation, and $\rightarrow$ is a total order on $T/\sim$. We can construct $\prec$ by putting the equivalence classes under $\sim$ in order according to $\rightarrow$, and arbitrarily ordering the elements of each equivalence class.

	The agent can traverse the tunnels in $T$ in the order described by $\prec$. This is a solution to the universal traversal problem.
	\end{proof}

	We run the algorithm in nondeterministic logarithmic space as follows. Begin with an NL algorithm that solves 2SAT, and assume the input is given in a format where we can check whether a clause $a\vee b$ is in the formula by checking a single bit for literals $a$ and $b$. For example, the input can be given as a matrix with a row and column for each literal. We run this nondeterministic 2SAT algorithm, except that whenever we would read a bit of the input, we perform a procedure to determine whether that clause is in the formula.

	Suppose the algorithm to solve universal traversal wants to know whether $a\vee b$ is in the formula. If $a$ and $b$ are both positive literals, we simply check whether they correspond to tunnels in the same gadget. If $a$ and $b$ have different signs, the clause is not in the formula. The interesting case is when $a=\neg x$ and $b=\neg y$ for tunnels $x$ and $y$, where we need to determine whether there is a path from the exit of $x$ to the entrance of $y$ or vice-versa.

	In this case, we nondeterministically guess whether the clause exists, and then check whether the guess was correct. If we guess it does exist, we run a coNL algorithm to verify that there is no path from the exit of $x$ to the entrance of $y$ or vice versa; this can be converted to an NL algorithm. If the verification succeeds, we proceed; if it fails, we halt and reject. Similarly, if we guess the clause does not exist, we run an NL algorithm to verify that there is such a path, proceeding on success and rejecting on failure.

	Consider the computation branches which have not rejected after this process. If the clause exists, the branch which attempted to verify it does not exist has entirely rejected, and the branch which attempted to verify it does exist has succeeded in at least one branch. So there is at least one continuing branch, and every such branch believes that the clause exists. Similarly if the clause does not exist, we end up with only branches which guessed that it does not exist. 

	At this point, we continue with the 2SAT algorithm, since every remaining branch knows the correct value for the input bit we have read.
\end{proof}

\begin{lemma}\label{lem:visit np-hard}
	Universal traversal with any one-state $k$-tunnel gadget with a directed edge is NP-hard when $k\geq3$.
\end{lemma}

\begin{proof}
	Let $G$ be any one-state $k$-tunnel gadget with a directed edge. We will only use three tunnels in each copy of $G$, at least one of which is directed. We begin by building the one-state gadget with three directed tunnels. To do this, we connect six copies of $G$ along three paths, such that each path goes through all six copies and the first and last tunnel on each path is directed. Six copies of $G$ is enough to supply these directed edges. The agent can only enter each tunnel of this construction from one side. When it does, it has no choice but to continue all the way through the construction, and in the process it uses all six gadgets involved. So this simulates the one-state gadget with three directed tunnels, and it suffices to show NP-hardness for this specific gadget.

	We prove NP-hardness by a reduction from 3SAT. Each clause becomes a copy of the gadget with three directed tunnels. There are a sequence of branches corresponding to variables, which go through tunnels in the gadgets corresponding to clauses containing the variable or its negation. At the branch corresponding to $x$, the agent must choose between a path which goes through the gadgets corresponding to clauses with $x$ and a path which goes through gadgets corresponding to clauses with $\neg x$. These paths merge before the branch for the next variable. The start location is at the first branch.

	A path through this system of gadgets is exactly an assignment for the formula, and the gadgets visited correspond to satisfied clauses. So it is possible to visit every gadget if and only if the formula is satisfiable.
\end{proof}

Combining Lemmas~\ref{lem:stateless np} through \ref{lem:visit np-hard}, we have Theorem~\ref{thm:stateless visit} characterizing the complexity of universal traversal for one-state $k$-tunnel gadgets.

\subsection{Reversible Deterministic Gadgets}\label{sec:unbounded visit}

In this subsection, we prove Theorem~\ref{thm:unbounded visit}, which shows that the complexity of the universal traversal problem for a reversible deterministic $k$-tunnel gadget is the same as the complexity of the reachability problem for that gadget.

\begin{theorem}\label{thm:unbounded visit}
	Let $G$ be a reversible deterministic $k$-tunnel gadget. Then universal traversal for $G$ is PSPACE-complete if $G$ has interacting tunnels, and is in NL otherwise.
\end{theorem}

\begin{proof}
	Suppose first that $G$ has no interacting tunnels. Then reachability for $G$ is in NL \cite{gadgets2}. It follows that the question of whether the agent can make a traversal in a target gadget is in NL, since we can solve the reachability question for each usable location of the gadget. To solve universal traversal for $G$, we check for each gadget whether the agent can use that gadget from the original configuration. If every gadget passes this check, we accept; otherwise we reject.

	The output of this algorithm is the AND of the outputs of polynomially many NL algorithms, so it runs in NL. The algorithm works because the agent can follow the path to a gadget, use that gadget, and then reverse its path back to the initial configuration. Doing this for each gadget in series, the agent can make a traversal in every gadget exactly when each gadget can be reached from the initial configuration.

	Now suppose $G$ has interacting tunnels. Containment in PSPACE is given by Lemma~\ref{lem:visiting in pspace}. To show PSPACE-hardness, we reduce from reachability for $G$, which is PSPACE-complete \cite{gadgets2}. By Lemma~2.5 in \cite{gadgets2}, $G$ simulates a 1-toggle (the 2-state gadget gadget with a single directed tunnel which flips when traversed). Furthermore, when the agent traverses the simulation of the 1-toggle, it visits every gadget in the simulation; this is clear from examining the proof of this lemma.

	Given an instance of the reachability problem for $G$, we modify it by adding a 1-toggle from the target location to each location in the system of gadgets. We also add a gadget at the target location which is usable if and only if the agent reaches the target location. This can clearly be done in polynomial time.

	If the agent can reach the target location, then it can travel along one of these 1-toggles, cross back and forth across a gadget, and return along the 1-toggle. In this way, the agent can make a traversal in every gadget, including those in the simulated 1-toggles and the gadget added at the target location. Conversely, in order to traverse every gadget, the agent most traverse the gadget at the target location, which requires being able to reach the target location in the original reachability instance. Thus the reachability instance is solvable if and only if the constructed universal traversal instance is solvable.
\end{proof}

\section{Gadget Reconfiguration}\label{sec:reconfiguration}
In this section we study the question of whether an agent has a series of moves after which the system of gadgets will be in some target configuration. In Section ~\ref{sec:ReconfigurationReversible} We show that for reversible, deterministic gadgets the reconfiguration problem is always PSPACE-complete if the reachability problem is PSPACE-complete and give an example of a reversible non-deterministic gadget with non-interacting tunnels for which the reconfiguration problem is PSPACE-complete (whereas reachability for any gadget with non-interacting tunnels is always in NL).  Section~\ref{sec:VerifiedGadgets} shows some methods for constructing new PSPACE-complete gadgets from known ones and shows the reconfiguration problem can be PSPACE-complete even when a gadget does not change traversability. Finally, in Section~\ref{sec:NPReDAG}, we show an interesting connection between reconfiguration problems and bounded reachability problems, expanding the classes of gadgets we know to be in NP. We also exhibit a gadget for which the reachability question is NP-complete but the reconfiguration question is in P.

\subsection{Reconfiguring Reversible Gadgets} \label{sec:ReconfigurationReversible}
In this section we first show that for any reversible gadget the reachability problem being PSPACE-complete implies the reconfiguration problem is also PSPACE-complete. We then exhibit a reversible, deterministic gadget with non-interacting tunnels for which the reconfiguration problem is PSPACE-complete showing an example where reconfiguration is a harder problem than reachability.

\begin{theorem}
	Reconfiguration motion planning is PSPACE-complete for any set of reversible gadgets for which reachability motion planning is PSPACE-complete.
\end{theorem}
\begin{proof}
	We use the same technique as the one used to show reconfiguration Nondeterministic Constraint Logic is PSPACE-complete \cite{hearn2005pspace}. First we take a hard instance of the reachability motion planning. Now at the target location we add a loop with a single gadget which permits a traversal which changes its state. For the reconfiguration problem, we set the target states of all but the newly added gadget to be the same as the initial states, and we set the target state of the added gadget to be one reachable by making a traversal in the loop. Since the gadgets in this system are all reversible, the agent can always take the inverse transitions that have been made so far to return the start location and restore the states of all the gadgets to their initial states. Thus if the agent can reach the added gadget, the agent will be able to traverse the gadget and then undo the rest of the path except for that final traversal. Since the agent must interact with the added gadget to achieve the desired reconfigured state, the agent must be able to reach the gadget. Thus the agent is able to solve the reconfiguration problem if and only if the agent could solve the reachability problem.
\end{proof}

\subsubsection{PSPACE-complete Reversible, Deterministic Gadget with Non-interacting Tunnels}
\label{sec:non-interacting PSPACE}

There are cases where the reconfiguration problem can be harder. Below we describe a reversible, deterministic gadget with non-interacting tunnels for which the reconfiguration problem is PSPACE-complete.

The \rdni{} is a reversible, deterministic, 12 state, two-tunnel gadget shown in Figure~\ref{fig:12StateGadget}. We will refer to states the right-top and right-middle states as the right leaf states and left-bottom and middle-bottom states as the bottom leaf states. We call the right leaf states and their two adjacent states the right square and similarly the bottom leaf states and their two adjacent states the bottom square. Notice that from some states a tunnel either only once in the same direction or potentially twice. Although going through one tunnel never changes the traversability of the other tunnel, it may change whether that tunnel can be traversed twice in a row in the same direction. 
To show PSPACE-completeness we will first show that a cooperative, multi-agent reconfiguration problem is hard by reduction from reconfiguration motion planning with a locking 2-toggle. We then show how we can augment that construction to allow a single agent to simulate the actions of all of the others in our multi-agent construction.

\begin{figure}
	\centering
	\includegraphics{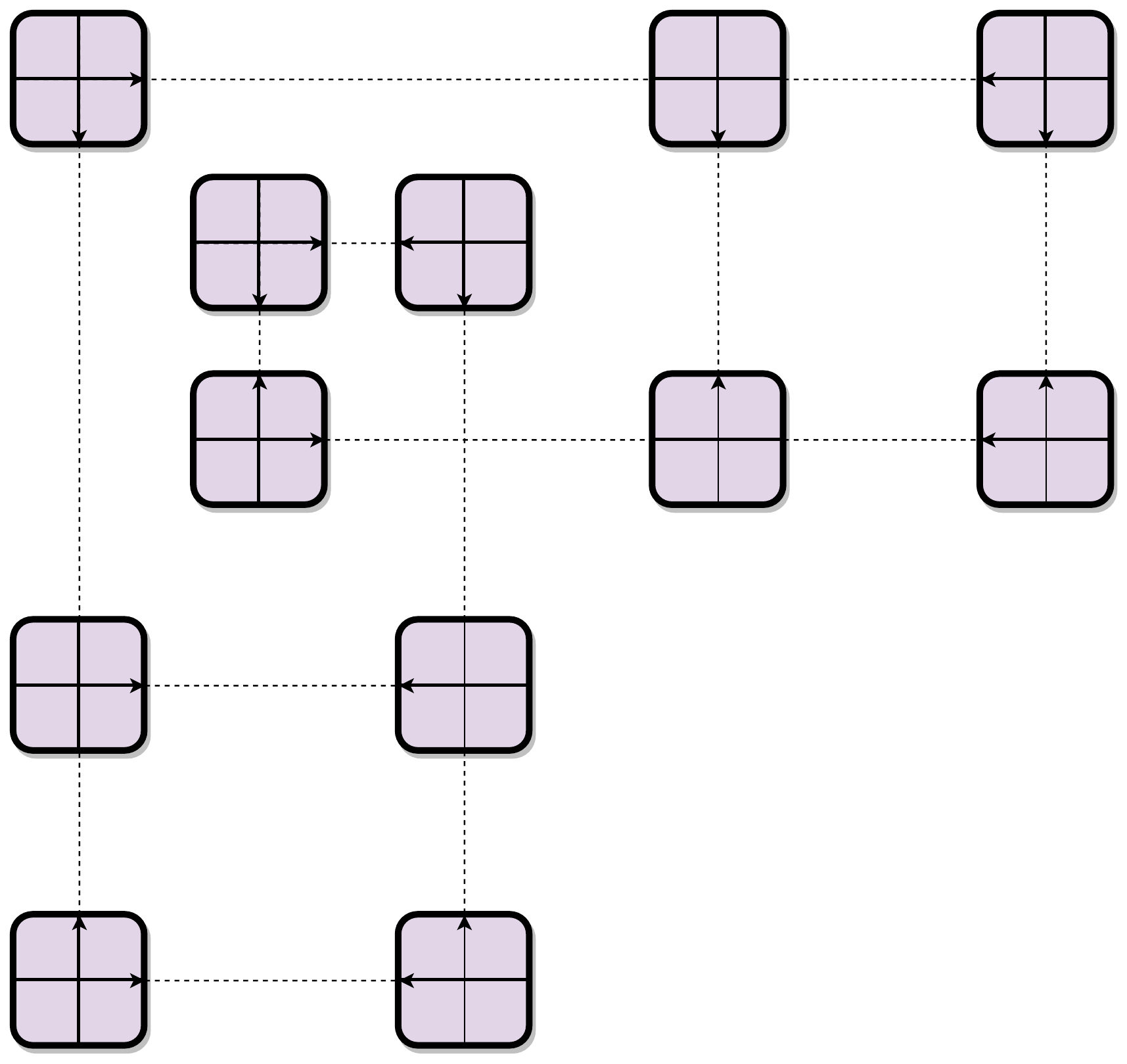}
	\caption{The state diagram of the \rdni}
	\label{fig:12StateGadget}
\end{figure}

\paragraph{Multi-agent 1-Toggle}
Recall a 1-toggle is a 2-state, 1-tunnel, reversible, deterministic gadget that allows a directed traversal in one direction in one state and the other direction in the other state. A regular 1-toggle can be easily constructed from the \rdni{} by taking a single tunnel in a leaf state. Instead we will build a gadget that does not allow individual agents through at all, but if it has an agent on either side of it, a third agent can use the gadget as though it were a 1-toggle. Our construction will only work as intended if there are three or fewer agents adjacent to the gadget at any point in time; however, these gadgets will only be used in a way that this condition remains satisfied.

\begin{figure}
	\centering
	\includegraphics{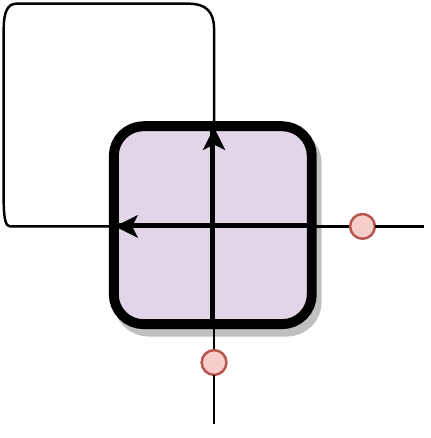}
	\caption{The multi-agent 1-toggle. The two helper agents are denoted by red dots.}
	\label{fig:multi-agent-1-toggle}
\end{figure}

To build our multi-agent 1-toggle we simply connect the tunnels together, as shown in Figure~\ref{fig:multi-agent-1-toggle} and consider the bottom-middle state to be the canonical configuration for the leftward pointing state of the 1-toggle, and the middle-right state to be the canonical configuration for the rightward pointing state. More configurations will need to be considered shortly, but we first describe the intended usage. In the bottom-middle state two agents can move up into the middle connection of the gadget, then the remaining agent can moves left joining them in the middle connection. The gadget is now in one of the upper left states. If one agent exits down, the other two can then exit to the right putting the gadget in the desired middle-right state with two agents on the right side. The transition in the other direction is symmetric.

Next we argue that the only ways agents can move through this gadget are equivalent to the intended usage. First, consider the case where there is no more than one agent on either side of the gadget. From the bottom middle state no right or down traversals can be made. Further, if no more than one up traversal is made then no more than one down traversal can be made, and the same is true for right and left. Thus the agents can put the gadget into four pairs of agent location and gadget state pairs, including only the bottom square states. Rather than just the middle-bottom state with one agent on each side, we actually consider these four agent and state pairs to be the rightward facing state of our multi-agent 1-toggle. Importantly these are the only reachable agent and state pairs and none of them have more than one agent on any side of the gadget.

Now consider the case where the gadget is rightward facing and there is a second agent on the right side. None of the previously reachable states will allow more than one left transition, and thus the second agent is unable to interact with the gadget.

Finally we're back to the case of a rightward pointing gadget with an extra agent on the bottom. In this state there can be no more downward traversals than upward traversals and there can be at most one more rightward traversal than leftward traversals. Thus we cannot move extra agents to the left side of the gadget and we can move at most one extra agent to the right side of the gadget. Further, after making two right traversals, there be at least two agents on the rightward side of the gadget and the gadget must be in a right leaf state. Thus, once there are two agents on the rightward side, we are in one of the right square states above with an extra agent. This is exactly the situation where the multi-agent one-toggle has changed state and allowed an agent to traverse it.

\paragraph{Multi-agent Locking 2-Toggle}

The multi-agent locking 2-toggle will be comprised of one \rdni{}, four multi-agent 1-toggles, and six helper agents. It will allow an additional agent to interact with it as though it were a locking 2-toggle. Two helper agents will be located in the horizontal and vertical connections next to the \rdni{}, and the other four agents will be external, each adjacent to one of the multi-agent 1-toggles. Note, these external four agents will be shared between gadgets rather than duplicated.

The canonical unlocked state is shown in Figure~\ref{fig:multi-agent-locking-2-toggle} with the \rdni{} in the upper-left state, the 1-toggles pointing right and down, and the internal agents in the left and top sides respectively. If a second agent comes in from the top, it is able to cross the first 1-toggle, both agents can move down through the \rdni{}, and then the two agents can allow one of them to cross the second 1-toggle. We consider this resulting state to be the canonical up locked state. The \rdni{} is in the lower left state, the 1-toggles are pointed right and up, and the internal helper agents are on the left and bottom. The right locked state and transition to it are symmetric.

\begin{figure}
	\centering
	\includegraphics{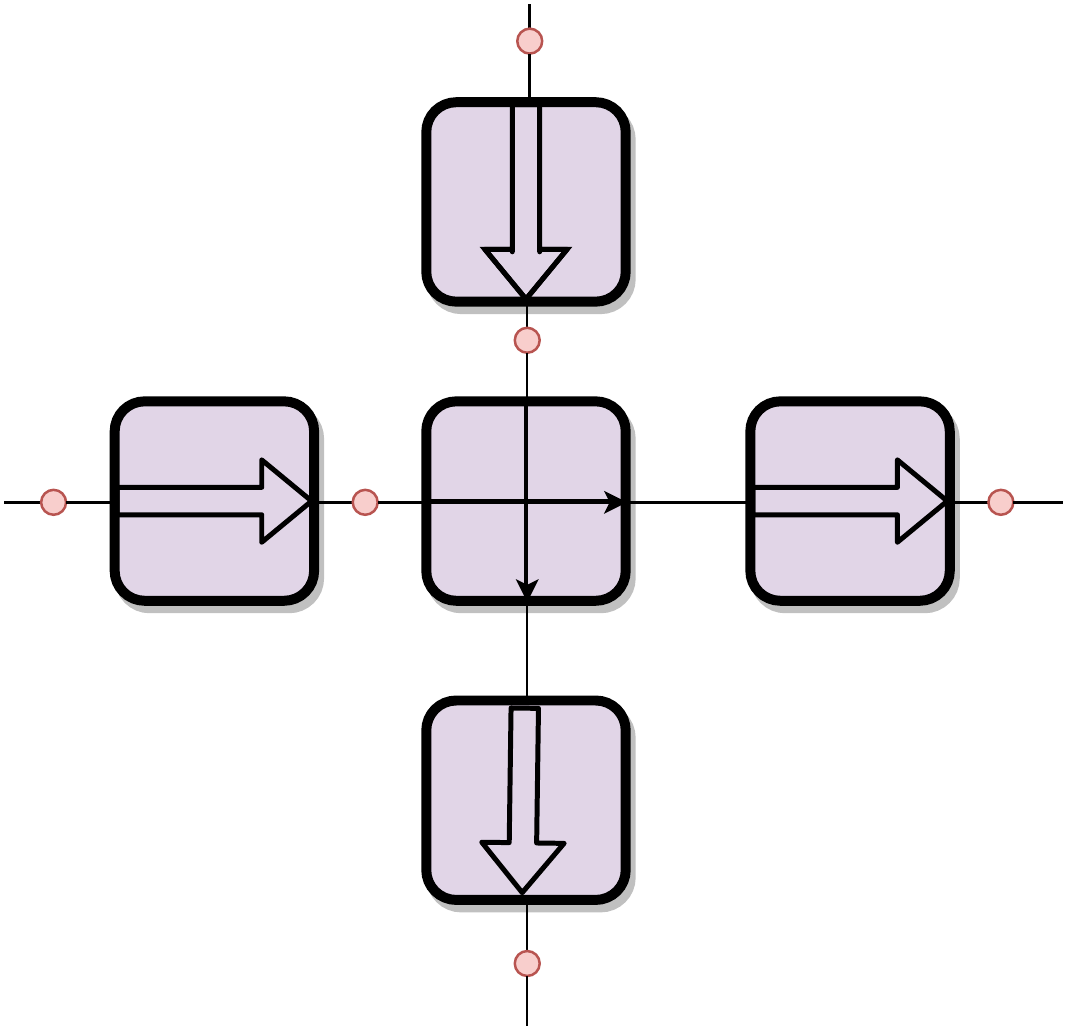}
	\caption{The multi-agent locking 2-toggle in the unlocked state.}
	\label{fig:multi-agent-locking-2-toggle}
\end{figure}

From the locked bottom state, if an agent arrives on the left, it could move through the 1-toggle, but only one agent would be able to pass the \rdni{}. Thus it would not be able to pass the second 1-toggle preventing a traversal of our gadget as desired. Notice in the locked states the internal helper agents will only be able to shift the state of the gadget over once, ensuring that they cannot move the gadget outside of the bottom or the right square states respective to being locked up or left. This is the property that ensures only one agent can cross the \rdni{} along the incorrect pathway while it is locked. The other traversals are prevented by the directionality of the multi-agent 1-toggles.

Now we must inspect the construction as a whole because our gadgets depend on never having more than two agents on any side and no more than three adjacent agents total. In this construction all pathways have two multi-agent 1-toggles between every \rdni{}. Since individual agents can freely cross the \rdni{}, let us imagine replacing them with simple connections. Now every multi-agent 1-toggle (except next to the start location) has exactly 1 agent on either side of it. There is only one additional agent in the entire construction, so from the properties of the multi-agent 1-toggles we know that we will never end up with more than one extra agent adjacent to any of the gadgets fulfilling our needed condition.

The states of the locking 2-toggles in the single agent 2-toggle reconfiguration can now be represented by the canonical agent location and gadget pairs in this multi-agent reduction. The movement of the agent is represented by the movement of the location containing two agents. Since the doubled agent can only move through the system of gadgets in the same way as the single agent in the original instance, and we've shown that the canonical pairs are reachable if the original instance is true, but are disjoint from reachable states which do not correspond to a valid traversal in the original instance we obtain PSPACE-hardness for the cooperative multi-agent reconfiguration problem with the \rdni{}.

\paragraph{Simulating Extra Agents}
Now we wish to simulate the multi-agent reduction with a single agent. Recall that we can directly build a (single agent) 1-toggle out of the \rdni{}. Also recall that our reduction ensured that no connection contained more than two agents at any given point in time. For each connection in the multi-agent instance, we connect two 1-toggles to that connection, each representing a potential agent. The other side of all the one toggles are connected together so that this central location has access to all of our simulated edges. If the multi-agent instance has an agent at some location, in the single-agent instance we direct a number of 1-toggles towards that connection for each agent there. All other 1-toggles are directed towards the central connection. Finally we start the agent in the central connection.

From the central connection, the agent is able to cross a 1-toggle 'instantiating' the agent it represents in the multi-agent problem. The agent is then in the same location and able to interact with original instance. If the agent then traverses a 1-toggle from some connection to the central connection, the flip in direction of that 1-toggle will allow the agent to return effectively 'remembering' the location of that agent in the multi-agent instance. Since the multi-agent problem never has more than two agents along the same connection, we need not worry about running out of 1-toggles to record the agent locations. If we pick one of the toggles to always flip when representing the presence of a single agent, we can directly map states of the gadgets onto pairs of gadget states and agent locations in the multi-agent instance, completing the reduction.

\subsection{Verified Gadgets and Shadow Gadgets}\label{sec:VerifiedGadgets}
In this section we will discuss a technique for generating hard gadgets. The main idea is constructing a gadget which behaves well when used like a hard gadget, but might also have other transitions which are allowed but put the gadget into some undesirable state.

First, we will pick some base gadget which we want to modify. Next we will add additional \emph{shadow states} to the gadget and additional transitions with the restriction that all newly added transitions must take the gadget to a shadow state. We call such a construction a \emph{shadow gadget} of the base gadget. This has the nice property that if the agent takes any transition that would not be allowed in the base gadget, then the gadget will always stay in a shadow state after that point.

\begin{theorem}\label{thm:shadow}
	Reconfiguration motion planning with a shadow gadget is at least as hard as reconfiguration motion planning with the base gadget.
\end{theorem}
\begin{proof}
	We simply take the hard instance for the problem with the base gadget and replace it with a shadow gadget. If a solution exists in the initial instance, it will still be a solution with the shadow gadgets. If the agent ever tries to take a transition in a shadow gadget that would not have been allowed in the original instance, that gadget will now be in a shadow state. Since no target state is a shadow state and all transitions from shadow states lead to shadow states such a path cannot be a solution.
\end{proof}

\begin{corollary}
	Reconfiguration motion planning is PSPACE-complete for some gadgets which never change their traversability.
\end{corollary}
\begin{proof}
	Take a gadget for which reconfiguration motion planning is PSPACE-complete, such as the 2-toggle. Now construct a shadow gadget with one shadow state that has transitions between all of the locations. Add transitions starting from every state and location and going to every other location and the shadow state. The resulting gadget always has available traversals from every location to every other location and thus never changes traversability. However, the reconfiguration problem is hard by Theorem~\ref{thm:shadow}.
\end{proof}

A \emph{verified gadget} is a shadow gadget with some additional structure. From a shadow gadget we add two or more locations, the \emph{verifying locations} to the gadget. We additionally may add \emph{verified states} which can only be reached by transitions from the added locations while the gadget is in normal states. We now add transitions among the verifying locations such that these locations can be connected in a series so there is always a traversal from the first to the last location if the gadget is in a normal state, and there is no such traversal if the gadget is in a shadow state.  We call this added traversal the \emph{verification traversal}. 

\begin{theorem}\label{thm:verified}
	Reachability motion planning with a verified gadget is at least as hard as reachability motion planning with the base gadget.
\end{theorem}
\begin{proof}
	We simply take the hard instance for the problem with the base gadget and replace it with a verified gadget and then make a path from the original target location through the verification traversals of all of the gadgets to a new target location. If a solution exists in the initial instance, then performing that solution will bring the agent to the start of the verification traversals and all of those traversals will be possible since the gadgets are all in normal states. If the agent ever tries to take a transition in a verifiable gadget that would not have been allowed in the original instance, that gadget will now be in a shadow state and at least one of the necessary verifiable traversals will now be possible.
\end{proof}

\paragraph{Monotonically Opening and Closing Gadgets}
 A \emph{monotonically opening} gadget is one in which the traversability of the gadget never decreases. That is to say for all states $t$ reachable from a given state $s$, and for all pairs of locations $a$ and $b$, if there is a transition from $a$ to $b$ in $s$ then there is a transition from $a$ to $b$ in $t$. A \emph{monotonically closing} gadget is one in which the traversability of the gadget never increases. That is to say, for all states $t$ reachable from a given state $s$, and for all pairs of locations $a$ and $b$, if there is no transition from $a$ to $b$ in $s$, then there is no transition from $a$ to $b$ in~$t$.

We now use verified gadgets to show that there are both monotonically opening and monotonically closing gadgets for which reachability is PSPACE-complete.  This is surprising because the number of changes of traversability in such a system of gadgets is bounded, so one might suspect such a class to fall in NP.

\begin{corollary}
	Reachability motion planning is PSPACE-complete even for monotonically closing gadgets.
\end{corollary}
\begin{proof}
	Take a gadget for which reconfiguration motion planning is PSPACE-complete, such as the 2-toggle. Now construct a shadow gadget with one shadow state that has transitions between all of the locations. Add transitions starting from every state and location and going to every other location and the shadow state. The resulting gadget always has available traversals from every location to every other location and thus never changes traversability. Next, we convert it into a verified gadget by adding a pair of locations $A$ and $B$ where there is a traversal between them if the gadget is in a normal state and no traversal if it is in a shadow state. This gadget now only removes traversals, but by Theorem~\ref{thm:verified} its reachability problem is PSPACE-complete.
\end{proof}

\begin{corollary}
	Reachability motion planning is PSPACE-complete even for monotonically opening gadgets.
\end{corollary}
\begin{proof}
	Take a gadget for which reconfiguration motion planning is PSPACE-complete, such as the 2-toggle. Now construct a shadow gadget with one shadow state that has transitions between all of the locations. Add transitions starting from every state and location and going to every other location and the shadow state. The resulting gadget always has available traversals from every location to every other location and thus never changes traversability. Next, we convert it into a verified gadget by adding two pairs of locations $A,B$ and $C,D$. There is a transition between $C$ and $D$ only if the gadget is in the verified state. There is additionally a transition between $A$ and $B$ from all normal states to the verified state, and also transitions between them from shadow states to shadow states. This gadget now only adds the traversal between $C$ and $D$ and never removes traversals, but by Theorem~\ref{thm:verified} its reachability problem is PSPACE-complete.
\end{proof}

\subsection{Reconfiguration and DAG-like Gadgets}
\label{sec:NPReDAG}
In \cite{gadgets2} DAG gadgets were studied as a naturally bounded class of gadgets and a broader notion LDAG was instigated in \cite{lynch2020framework}. In particular DAG gadgets have state transition graphs which are directed acyclic graphs (dags) and LDAG gadgets have state transition graphs which are dags with the addition of self-loops at some of the vertices. We now consider a generalized class of gadgets and describe cases in which the reachability question remains in NP. 

We call a gadget \emph{$F$-DAG-like} if its state graph can be decomposed into disjoint subgraphs for which those subgraphs are from some family of gadgets $F$ and all transitions between these subgraphs form a DAG. We call these transitions between the subgraphs \emph{DAG-like transitions}. In this case LDAG gadgets are $F$-DAG-like with some family of single state gadgets and all of the state changing transitions are the DAG-like transitions.

With this notion, one may wonder what gadgets can be used in an $F$-DAG-like gadget and have the resulting gadget still be in NP. We initially believed this would be true for gadgets with non-interacting tunnels, however, below we give an example of such a gadget where the reachability question is PSPACE-complete. We then show that if $F$ is a family of gadgets for which the reconfiguration problem is in NP, then the reconfiguration and reachability problems for $F$ and for $F$-DAG-like gadgets are also in NP. We call $F$-DAG-like gadgets where $F$ is a family of gadgets for which the reconfiguration problem is in NP \emph{NPReDAG gadgets}.

\begin{theorem}
	Reconfiguration motion planning with NPReDAG gadgets is in NP.
\end{theorem}
\begin{proof}
	We give the following certificate for 1-player motion planning with an NPReDAG gadget. We list all of the DAG-like transitions taken in the solution and the states of all of the gadgets before and after the transition. Further, for each pair of adjacent DAG-like transitions we imagine the reconfiguration problem on the system of gadgets which is only comprised by the reconfigurable super-node gadgets and takes this system from the state after the last DAG-like transition to the state before the next DAG-like transition. This problem is solvable in NP by definition, so we provide each of these certificates. The verifier can now check in polynomial time that the final state is the target state, that the polynomial many DAG-like transitions are valid transitions and take the given pre-transition state to the post-transition state, and that the (polynomially many) portions of the path between the DAG-like transitions have some valid path performing that reconfiguration. 
\end{proof}

\begin{theorem}
	Reachability motion planning is in NP for gadgets where reconfiguration motion planning is in NP.
\end{theorem}
\begin{proof}
	We now essentially want to `guess' the final configuration that the system will be in when the agent solves the reconfiguration problem and then solve the reconfiguration problem. However, this strategy also needs to verify that the agent actually reaches the target location. To do this first we take the reachability instance and add at the target location a loop with a gadget that has access to a transition with a state change. If there is none, then both the reconfiguration and reachability problems are trivially in NL. To be able to change the state of the added gadget, an agent must have reached the location of the loop. Thus we will take as a certificate a final configuration of the system of gadgets which has the added gadget in a different state, as well as the certificate for the reconfiguration problem from the initial state to this new target state.
\end{proof}

\begin{corollary}
	Reachability motion planning with NPReDAG gadgets is in NP.
\end{corollary}

\subsection{Reconfiguration Can Be Easier}
\label{sec: Reconfiguration Easier}

In this section we introduce the Labeled Two-Tunnel Single-Use gadget for which the reachability question is harder than the reconfiguration problem. Figure~\ref{fig:reconfiguration-easier} shows the Labeled Two-Tunnel Single-Use gadget. It is a DAG gadget where going through either tunnel of the gadget closes both of them; however, the states are distinguished based on which tunnel was traversed. This is a DAG gadget with a forced distant door closing, so it is NP-complete by Theorem~22 in \cite{gadgets2}. We now give a polynomial time algorithm for the reconfiguration problem.

\begin{figure}
	\centering
	\includegraphics{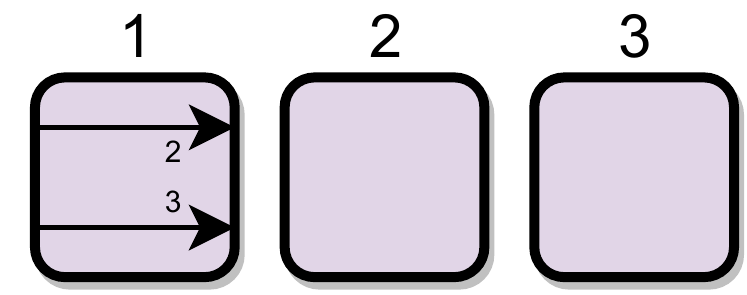}
	\caption{The Labeled Two-Tunnel Single-Use gadget.}
	\label{fig:reconfiguration-easier}
\end{figure}

\begin{theorem}
	Reconfiguration motion planning with the Labeled Two-Tunnel Single-Use gadget is in P.
\end{theorem}
\begin{proof}
	We will call states $2$ and $3$ \emph{terminal states}. Now let us consider what the initial and final configurations of the gadgets can look like. If the initial state is terminal, the gadget cannot be traversed. Similarly, if the initial and final configuration are both state $1$, then the gadget cannot have been traversed since there is no way to return the gadget to state $1$ after traversal. Thus the only case we need to consider is starting in state $1$ and ending in a terminal state. In this case, the labeling of the target configuration tells us which of the two tunnels must have been traversed to reach that state. We can thus construct the graph which uses only those tunnels and ask whether there is a path which traverses them all exactly once. Since this is exactly checking for the existence of an Eulerian path in a graph, we can solve it in polynomial time.
\end{proof}

It would be interesting to have an example of a gadget which has a different traversability in every state so that the easiness of such a reconfiguration problem would not be using a degeneracy which is indistinguishable in the reachability problem.

\bibliographystyle{plain}
\bibliography{refs}

\begin{thebibliography}{1}

\bibitem{akitaya2021characterizing}
Hugo~A Akitaya, Erik~D Demaine, Andrei Gonczi, Dylan~H Hendrickson, Adam
  Hesterberg, Matias Korman, Oliver Korten, Jayson Lynch, Irene Parada, and
  Vera Sacrist{\'a}n.
\newblock Characterizing universal reconfigurability of modular pivoting
  robots.
\newblock In {\em 37th International Symposium on Computational Geometry},
  2021.

\bibitem{balanza2019full}
Jose Balanza-Martinez, Austin Luchsinger, David Caballero, Rene Reyes, Angel~A
  Cantu, Robert Schweller, Luis~Angel Garcia, and Tim Wylie.
\newblock Full tilt: {U}niversal constructors for general shapes with uniform
  external forces.
\newblock In {\em Proceedings of the Thirtieth Annual ACM-SIAM Symposium on
  Discrete Algorithms}, pages 2689--2708. SIAM, 2019.

\bibitem{gadgets}
Erik~D. Demaine, Isaac Grosof, Jayson Lynch, and Mikhail Rudoy.
\newblock Computational complexity of motion planning of a robot through simple
  gadgets.
\newblock In {\em Proceedings of the 9th International Conference on Fun with
  Algorithms (FUN 2018)}, pages 18:1--18:21, La Maddalena, Italy, June 2018.

\bibitem{gadgets2}
Erik~D. Demaine, Dylan Hendrickson, and Jayson Lynch.
\newblock Toward a general theory of motion planning complexity: Characterizing
  which gadgets make games hard.
\newblock In {\em Proceedings of the 11th Conference on Innovations in
  Theoretical Computer Science (ITCS 2020)}, pages 62:1--62:42, Seattle,
  Washington, January 2020.

\bibitem{hearn2005pspace}
Robert~A. Hearn and Erik~D. Demaine.
\newblock {PSPACE}-completeness of sliding-block puzzles and other problems
  through the nondeterministic constraint logic model of computation.
\newblock {\em Theoretical Computer Science}, 343(1--2):72--96, 2005.

\bibitem{nl=conl}
Neil Immerman.
\newblock Nondeterministic space is closed under complementation.
\newblock {\em SIAM Journal on Computing}, 17(5):935--938, 1988.

\bibitem{lynch2020framework}
Jayson Lynch.
\newblock {\em A framework for proving the computational intractability of
  motion planning problems}.
\newblock PhD thesis, Massachusetts Institute of Technology, 2020.

\bibitem{SAVITCH}
Walter~J. Savitch.
\newblock Relationships between nondeterministic and deterministic tape
  complexities.
\newblock {\em Journal of Computer and System Sciences}, 4(2):177--192, 1970.

\bibitem{connectivity}
Avi Wigderson.
\newblock The complexity of graph connectivity.
\newblock In Ivan~M. Havel and V{\'a}clav Koubek, editors, {\em Proceedings of
  the 17th International Symposium on Mathematical Foundations of Computer
  Science (MFCS 1992)}, pages 112--132, Prague, Czechoslovakia, 1992.

\end{thebibliography}
\end{document}